\documentclass[conference]{IEEEtran}

\usepackage[english]{babel}
\usepackage{latexsym}
\usepackage{amssymb,amsmath}
\usepackage{amsthm}
\usepackage{cite}    
\usepackage{ifthen}
\usepackage{verbatim}  
\usepackage{xspace}
\usepackage{paralist}
\usepackage{booktabs}
\usepackage[colorlinks,final,citecolor=blue]{hyperref}

\usepackage{tikz}
\usetikzlibrary{%
  arrows,%
  positioning,%
  decorations.pathmorphing,%
  decorations.pathreplacing,%
  automata,%
}

\newenvironment{transducer}[1][]
	{\begin{center}
	\begin{tikzpicture}[font=\footnotesize,shorten >=1pt,node distance=2.75cm,
	on grid,>=stealth',initial text=,every node/.style={align=center},
	every state/.style={inner sep=1pt},
	every path/.style={->,bend angle=20},#1]}
	{\end{tikzpicture}\end{center}}

\newcommand\pssn{\par\smallskip\noindent}

\newcommand\pnsi{\par\indent}
\newcommand\pnsn{\par\noindent}
 
\newlength{\algoindent}
\theoremstyle{plain}
\newtheorem{theorem}{Theorem}[]

\newtheorem{lemma}[theorem]{Lemma}

\theoremstyle{definition}
\newtheorem{definition}[theorem]{Definition}

\theoremstyle{remark}
\newtheorem{remark}[theorem]{Remark}

\renewcommand{\phi}{\varphi}
\renewcommand{\epsilon}{\varepsilon}

\newcommand{\wrt}{w.\,r.\,t.\xspace}
\newcommand\nextword{{\rm\texttt{nonMax}}}
\newcommand\makecode{{\rm\texttt{makeCode}}}
\newcommand\pickfrom{{\rm\texttt{pickFrom}}}

\newcommand\none{{\rm\texttt{None}}}
\newcommand\ff{\ensuremath{f}}  
\newcommand\ef{\ensuremath{\epsilon}}  
\newcommand\fp{95\%}
\newcommand\fv{0.95}    
\newcommand\ep{5\%}
\newcommand\ev{0.05}    
\newcommand\nv{2\,000}  
\newcommand\prob[1]{{\mathtt{Pr}\xspace\Big[#1\Big]}}  

\newcommand{\mlt}{\ensuremath{<}\xspace}
\newcommand{\mle}{\ensuremath{\le}\xspace}

\newcommand{\sse}{\subseteq}

\newcommand{\sm}{\setminus}

\newcommand{\e}{\epsilon}
\newcommand{\ew}{\epsilon}          
\newcommand\al{A}        
\newcommand\lblock{\al^\ell}
\newcommand\symb{a}           

\newcommand\aut{\ensuremath{\mathbf{a}}}   
\newcommand\autb{\ensuremath{\mathbf{b}}}   
\newcommand\autv{\ensuremath{\mathbf{v}}}   
\newcommand\autc{\ensuremath{\mathbf{c}}}   
\newcommand\autd{\ensuremath{\mathbf{d}}}   
\newcommand\autm{\ensuremath{\mathbf{m}}}   
\DeclareMathOperator{\opL}{L}
\newcommand\lang[1]{\ensuremath\opL(#1)} 
\DeclareMathOperator{\opC}{C}
\newcommand\code[1]{\ensuremath\opC(#1)} 
\newcommand{\trt}{\mathbf{t}}   
\newcommand\trs{\mathbf{s}}      
\newcommand{\trti}{\trt^{-1}}   
\newcommand\chann{\ensuremath{\mathbf{er}}\xspace}    
\newcommand\err{\chann}      
\newcommand\erri{\err^{-1}}      
\newcommand\chid[1]{\mathtt{id}_{#1}} 
\newcommand\chbsid[1]{\mathtt{bsid}_{#1}} 
\newcommand\chsub[1]{\mathtt{sub}_{#1}} 
\newcommand\chdel[1]{\mathtt{del}_{#1}} 
\newcommand\chsegd[1]{\mathtt{segd}_{#1}} 
\newcommand\chins[1]{\mathtt{ins}_{#1}} 
\newcommand\chov{\mathtt{ov}} 
\newcommand\stoch{\mathbf{sc}}      
\DeclareMathOperator{\opMI}{maxind}
\newcommand\mind[2]{\ensuremath{\opMI(#1,#2)}}

\newcommand\filter{\rhd}   
\newcommand{\fullblock}{\ensuremath{\mathrm{FULLBLOCK}}\xspace}
\newcommand{\maxed}{\ensuremath{\mathrm{MAXED}}\xspace}
\newcommand{\ofw}{\ensuremath{\mathrm{OF}}\xspace}

\begin{document}
\title{Channels with Synchronization/Substitution Errors\\
       and  Computation of Error Control Codes}

\author{
\IEEEauthorblockN{Stavros Konstantinidis}
\IEEEauthorblockA{Department of Mathematics
and Computing Science\\
Saint Mary's University\\
Halifax, Nova Scotia, Canada\\
Email: s.konstantinidis@smu.ca\\
(research supported by NSERC)}
\and
\IEEEauthorblockN{Nelma Moreira and Rog{\'e}rio Reis}
\IEEEauthorblockA{CMUP \& DCC\\
Faculdade de Ci{\^e}ncias da Universidade do Porto\\
Rua do Campo Alegre, 4169–007 Porto Portugal\\
Email: \{nam,rvr\}@dcc.fc.up.pt\\
(research supported by FCT project UID/MAT/00144/2013)}
}

\maketitle

\begin{abstract}
We introduce the concept of an \ff-maximal error-detecting block code, for some
parameter \ff{} between 0 and 1, in order
to formalize the situation where a block code is close to maximal with respect to being error-detecting. Our motivation for this is that constructing a maximal error-detecting code is a computationally hard problem. 
We present a randomized algorithm that takes as input 
two positive
integers $N,\ell$, a probability value \ff, and a specification of the 
errors permitted in some application, and generates an error-detecting, or error-correcting, block code 
having up to $N$ codewords of length $\ell$. 
If the algorithm finds less than $N$ codewords, then those codewords constitute a code that is \ff-maximal with high probability. 
The error specification (also called channel)
is modelled as a transducer,
which allows one to model  any rational
combination of substitution and synchronization errors.
We also present some elements of our implementation of various error-detecting properties and their associated methods. Then, we  show several tests of the implemented randomized algorithm on various channels.
A methodological contribution is the
presentation of how various desirable error combinations 
can be expressed formally and processed
algorithmically.
\end{abstract}

\IEEEpeerreviewmaketitle

\section{Introduction}\label{sec:intro}
We  consider block codes $C$, that is, sets of words
of the same length $\ell$, for some integer $\ell>0$.
The elements of $C$ are called \emph{codewords} or 
\emph{$C$-words.}
We use $\al$ to denote the alphabet used for making words and
\[
\al^\ell = \mbox{the set of all words of length $\ell$.}
\]
Our typical alphabet will be the binary one $\{0,1\}$.
We shall use the variables $u,v,w,x,y,z$ 
to denote words over $\al$ (not necessarily in $C$). 
The \emph{empty word} is denoted by $\ew$.
\pnsi
We also consider error specifications $\err$, which we call
\emph{combinatorial channels}, or simply 
\emph{channels}. A channel
$\err$ specifies, for each allowed input word $x$, the set $\err(x)$
of all possible output words. We assume that error-free
communication is always possible, so $x\in\err(x)$.
On the other hand, if $y\in\err(x)$ and $y\not=x$ then
the channel introduces errors into $x$.
Informally, a block code $C$ is \err-detecting if the channel cannot turn a given $C$-word
into a \emph{different} $C$-word. It is \err-correcting
if the channel cannot turn two \emph{different} $C$-words into the same word. 
\pnsi
In Section~\ref{sec:channels}, we make the above concepts 
mathematically precise, and show how known examples of combinatorial channels can be defined formally so that
they can be used as input to algorithms.
In Section~\ref{sec:main}, we present two randomized
algorithms: the first one decides
(up to a certain degree of confidence) whether a
given block code $C$ is maximal \err-detecting for a
given channel $\err$. The second algorithm is given a
channel $\err$, an \err-detecting 
block code $C\subseteq\al^\ell$ (which could be empty), and
  integer $N>0$,
and attempts to add to $C$ $N$ new words 
of length $\ell$ resulting into a new \err-detecting
code. If less than $N$ words get added then either
the new code is \fp-maximal or the chance that
a randomly chosen word can be added is less than \ep.
Our motivation for considering a randomized algorithm  is that embedding a given \err-detecting block code $C$ into a maximal \err-detecting block code is a computationally hard problem---this is shown in Section~\ref{sec:coNP}.
In Section~\ref{sec:implem}, we discuss briefly some
capabilities of the 
new module \texttt{codes.py}
in the open source software package FAdo \cite{Fado,AAAMR:2009}
and we discuss some tests of the randomized 
algorithms on various channels.
In Section~\ref{sec:further}, we discuss a few more
points on channel modelling and conclude with
directions for future research.

We note that, while there are various algorithms for computing
error-control codes, to our knowledge these work
for specific channels and implementations are generally not 
open source.

\section{Channels and Error Control Codes}\label{sec:channels}
We need a mathematical
model for channels that is useful
for answering \emph{algorithmic} questions pertaining to
error control codes. While many models of channels and codes
for substitution-type errors use a rich 
set of mathematical structures, this is not the case for 
channels involving synchronization errors \cite{MBT:2010}. 
We believe the
appropriate model for our purposes is that of a transducer.
We note that transducers have  been defined as early as in 
\cite{ShaWea:1949}, and are a  powerful computational 
tool for
processing sets of words---see \cite{Be:1979} and 
pg 41--110 of \cite{FLhandbookI}.
\begin{definition}\label{def:transd}
A  \emph{transducer} is a 5-tuple\footnote{The general definition of transducer allows two alphabets: the input and the output alphabet. Here, however, we assume that both alphabets are the same.}
$\trt=(S,\al,I,T,F)$ 
such that $\al$ 
is the alphabet, $S$ is the finite  set of states, $I\sse S$ is the set of initial states, $F\sse S$ is the set of final states,
and $T$ is the finite set of transitions. Each
transition is a 4-tuple $(s_i, x_i/y_i, t_i)$,
where $s_i,t_i\in S$ and $x_i,y_i$ are words over $\al$. 
The word $x_i$ is the \emph{input label} and the word $y_i$ 
is the \emph{output label} of the transition. For two words $x,y$
we write $y\in\trt(x)$ to mean that $y$ is a possible output
of $\trt$ when $x$ is used as input. More precisely, there
is a sequence
\[
(s_0,x_1/y_1,s_1),\>(s_1,x_2/y_2,s_2),\ldots,\> (s_{n-1},x_n/y_n,s_n)
\]
of transitions such that $s_0\in I$, $s_n\in F$,
$x=x_1\cdots x_n$ and $y=y_1\cdots y_n$. 
The \emph{relation} $R(\trt)$ \emph{realized} by $\trt$ is the set of word pairs $(x,y)$ such that $y\in\trt(x)$. 
A relation $\rho\sse\al^*\times\al^*$ is called \emph{rational} if it is realized by a transducer.
If every input and every output label of $\trt$ is in $\al\cup\{\ew\}$, then we say that $\trt$ is in \emph{standard form}.
The \textit{domain} of the transducer $\trt$ is the set of words $x$ such that $\trt(x)\not=\emptyset$. The transducer is called
\textit{input-preserving} if $x\in\trt(x)$, for all words $x$ in the domain of $\trt$.
The \emph{inverse}  of $\trt$, denoted by $\trti$, is the transducer that is
simply obtained by
making a copy of 
$\trt$ and changing each transition
$(s,x/y,t)$ to $(s,y/x,t)$. Then
\[
x\in\trti(y)\>\mbox{ if and only if }\> y\in\trt(x).
\]
\end{definition}
We note that every transducer can be converted (in linear time) to one in standard form realizing the same relation.
\pnsi 
In our objective to model channels $\err$ as transducers, 
we require that a transducer $\err$ is a channel if it allows error-free communication, that is, 
$\err$ is input-preserving.
\begin{definition}\label{def:chann}
An \emph{error specification} \err{} is an input-preserving transducer.
The \emph{(combinatorial) channel} specified by $\err$ is 
$R(\err)$, that is, the relation realized by \err. For the purposes of this paper, however, we simply identify the concept of channel with that of error specification.
\end{definition}
A piece of notation that is useful in this work is
the following, where $W$ is any set of words,
\begin{equation}\label{eq:erronset}
\err(W)\>=\>\bigcup_{w\in W}\err(w)
\end{equation}
Thus, $\err(W)$ is the set of all possible outputs of $\err$
when the input is any word from $W$. For example, if $\err=\chid1$ = the channel that 
allows up to 1 symbol to be deleted or inserted in the input word, then $\chid1(\{00, 11\})\>=\>$
\[
\{00,0,000,100,010,001,11,1,011,101,110,111\}.
\]
Fig.~\ref{fig:ed} considers  examples of
channels that have been defined
in past research when designing error control codes.
Here these channels are shown as transducers,
which can be used as inputs to algorithms for 
computing error
control codes. 
For the channel $\chsub2$, we have 
$
00101\in\chsub2(00000)
$
because on input 00000, the channel $\chsub2$ can read 
the first two input 0's at state $s$
and output 0, 0; 
then, still at state $s$, read the 3rd 0 and output 1 and go to state $t_1$; etc.
\begin{figure}[ht!]
\begin{transducer}
	\node [state,initial,accepting] (q0) {$s$};
	\node [node distance=0.75cm,left=of q0,anchor=east] 
	      {$\chsub2=\>$};
	\node [state,accepting,right of=q0] (q1) {$t_1$};
	\node [state,accepting,right of=q1] (q1b) {$t_2$};
	\node [node distance=2.25cm, state,initial,accepting,below of=q0] (q2) {$s$};
	\node [node distance=0.75cm,left=of q2,anchor=east] 
	      {$\chid2=\>$};
	\node [state,accepting,right of=q2] (q3) {$t_1$};
	\node [state,accepting,right of=q3] (q3b) {$t_2$};
	\node [node distance=2.25cm, state,initial,accepting,below of=q2] (q4) {$s$};
	\node [node distance=0.75cm,left=of q4,anchor=east] 
	      {$\chdel1=\>$};
	\node [state,right of=q4] (q5) {$r$};
	\node [state,accepting,right of=q5] (q5b) {$t$};
	\node [node distance=2.25cm,state,initial,accepting,accepting,below of=q4] (q6) {$s$};
	\node [node distance=0.75cm,left=of q6,anchor=east] 
	      {$\chins1=\>$};
	\node [state,right of=q6] (q7) {$r$};
	\node [state,accepting,right of=q7] (q7b) {$t$};
	\path (q0) edge [loop above] node [above] {$0/0,\,1/1$} ()
		(q0) edge node [above] {$0/1$} (q1)
		(q0) edge node [below] {$1/0$} (q1)
		(q1) edge [loop above] node [above] {$$0/0,\,1/1$$} ()

		(q1) edge node [above] {$0/1$} (q1b)
		(q1) edge node [below] {$1/0$} (q1b)
		(q1b) edge [loop above] node [above] {$$0/0,\,1/1$$} ()

		(q2) edge node [above] {$\symb/\ew$} (q3)
		(q2) edge node [below] {$\ew/\symb$} (q3)
		(q3) edge node [above] {$\symb/\ew$} (q3b)
		(q3) edge node [below] {$\ew/\symb$} (q3b)
		(q3b) edge [loop above] node [above] {$a/a$} ()
		(q2) edge [loop above] node [above] {$\symb/\symb$} ()
		(q3) edge [loop above] node [above] {$\symb/\symb$} ()

		(q4) edge node [above] {$\symb/\ew$} (q5)
		(q5) edge node [above] {$\ew/\symb$} (q5b)
		(q4) edge [loop above] node [above] {$\symb/\symb$} ()
		(q5) edge [loop above] node [above] {$\symb/\symb$} ()
	    (q6) edge node [above] {$\ew/\symb$} (q7)
		(q7) edge node [above] {$\symb/\ew$} (q7b)
		(q6) edge [loop above] node [above] {$\symb/\symb$} ()
		(q7) edge [loop above] node [above] {$\symb/\symb$} ();
\end{transducer}
\begin{center}
\caption{
Examples of (combinatorial) channels: $\chsub2,\chid2,\chdel1,\chins1$. 
\underline{Notation}: A short arrow with no label
points to an initial state (e.g., state $s$), and a double
line indicates a final state (e.g., state $t$). 
An arrow with label $\symb/\symb$ represents multiple 
transitions, each with label $\symb/\symb$, for
 $\symb\in\al$;
and similarly for an arrow with label  $\symb/\ew$---recall, 
$\ew$ = empty word.
Two or more labels on one arrow from some state $p$ to some state $q$ represent
multiple transitions between $p$ and $q$ having these labels. 
\underline{Channel  $\chsub2$}: 
uses the binary alphabet $\{0,1\}$. On input $x$, $\chsub2$ 
outputs $x$, or any word that results by performing one or two substitutions in $x$. The latter case
is when $\chsub2$ takes the transition $(s,0/1,t_1)$ or
$(s,1/0,t_1)$, corresponding to one error, and then possibly
$(t_1,0/1,t_2)$ or
$(t_1,1/0,t_2)$, corresponding to a second error.
A block code $C$ is $\chsub2$-detecting iff 
the min. Hamming distance of $C$ is $>2$.
\underline{Channel  $\chid2$}: alphabet $\al$ not specified.
On input $x$,  $\chid2$ 
outputs a word that results by inserting and/or deleting 
at most 2 symbols in $x$. A block code $C$ is $\chid2$-detecting iff 
the min. Levenshtein distance of $C$ is $>2$ \cite{Levenshtein:66:en}.
\underline{Channels  $\chdel1,\chins1$}: considered in \cite{PAF:2013}, here alphabet 
$\al$ not specified.
On input $x$,  $\chdel1$ 
outputs either $x$, or any word that results by deleting
exactly one symbol in $x$ and then inserting a symbol at the end of~$x$.}
\label{fig:ed}
\end{center}
\end{figure}
\pnsi
The concepts of error-detection and -correction mentioned
 in the introduction are  phrased below
more rigorously.
\begin{definition}\label{def:errcontrol}
Let $C$ be a block code and let $\err$ be a channel.
We say that $C$ is \emph{$\err$-detecting} if
\[
v\in C,\> w\in C \>\mbox{ and }\> w\in\err(v)\>\>\mbox{imply} \>\>
v=w.
\]
We say that $C$ is \emph{$\err$-correcting} if
\[
v\in C,\> w\in C \>\mbox{ and }\> x\in\err(v)\cap\err(w)\>\>\mbox{ imply} \>\>
v=w.
\]
An $\err$-detecting block code $C$ is called \emph{maximal
$\err$-detecting} if $C\cup\{w\}$ is not $\err$-detecting for any word of length $\ell$ that is not in $C$. The concept
of a \emph{maximal
$\err$-correcting}  code is similar.
\end{definition}
%
%
From a logical point of view 
(see Lemma~\ref{lem:edvsec} below) 
error-detection
subsumes the concept of error-correction. This connection is stated
already in \cite{KonSil:2008} but without making use of it there.
Here we add the fact that maximal error-detection subsumes
maximal error-correction.  Due to this observation,
in this paper  we focus only on 
error-detecting codes. 
\par
\underline{Note}: The operation `$\circ$' between two 
transducers  $\trt$ and $\trs$ is called 
\emph{composition} and returns a new 
transducer $\trs\circ\trt$ such that
$
z\in(\trs\circ\trt)(x)\>\>\mbox{ if and only if }\>\>
y\in\trt(x)\>\mbox{and}\>z\in\trs(y),\mbox{for some $y$}.
$
\begin{lemma}\label{lem:edvsec}
Let $C\sse\al^\ell$ be a block code  
and  $\err$ be a channel. Then
$C$ is $\err$-correcting if
and only if it is $(\erri\circ\err)$-detecting. 
Moreover, $C$ is maximal $\err$-correcting if
and only if it is maximal $(\erri\circ\err)$-detecting.
\end{lemma}
\begin{proof}
The first statement is already in \cite{KonSil:2008}. 
For the second statement,
first assume that $C$ is maximal $\err$-correcting 
and consider any word $w\in\al^\ell\setminus C$. 
If $C\cup\{w\}$ were $(\err\circ\erri)$-detecting then
$C\cup\{w\}$ would also be $\err$-correcting and, hence,
$C$ would be non-maximal; a contradiction. Thus, 
$C$ must be maximal $(\err\circ\erri)$-detecting. The converse
can be shown analogously.
\end{proof}
The operation `$\lor$' between any two transdcucers $\trt$ and $\trs$
is obtained by simply taking the union of their five corresponding 
components (states, alphabet, initial states, 
transitions, final states) after a renaming, if necessary, of the states
such that the two transdcucers have no states in common.
Then
\[
(\trt\lor\trs)(x)\>\> = \>\> \trt(x)\>\cup\>\trs(x).
\]
Let  $\err$ be a channel, let $C\sse\al^\ell$ be an \err-detecting block code,  and let $w\in\lblock\sm C$.
In \cite{DudKon:2012}, the authors show that 
\begin{equation}\label{eq:addword}
\mbox{
$C\cup\{w\}$ is $\err$-detecting iff $w\notin(\err\lor\erri)(C)$.
}
\end{equation}

\begin{definition}\label{def:addword}
Let $C\sse\lblock$ be an $\err$-detecting block code.
We say that  \emph{a word $w$ can be added into $C$} if $w\notin(\err\lor\erri)(C)$.
\end{definition}

Statement~(\ref{eq:addword}) above implies that 
\begin{equation}\label{eq:maxed}
\mbox{
$C$ is maximal $\err$-detecting \;iff\;
$\al^\ell\;\setminus\;(\err\lor\erri)(C)\;=\;\emptyset$.}
\end{equation}
\begin{definition}\label{def:maxind}
The \emph{maximality index} of a block code 
$C\sse\al^\ell$ \wrt 
 a channel $\err$ is the quantity
\[
\mind{C}{\err} = \frac{|\al^\ell\cap(\err\lor\erri)(C)|}{|\al^\ell|}.
\]
Let $\ff$ be a real number in $[0,1]$.
An \err-detecting block code $C$ is called \emph{\ff-maximal 
$\chann$-detecting} if $\mind{C}{\err}\ge\ff$.
\end{definition}
The  maximality index of $C$ is the proportion of 
the `used up' words of length $\ell$ over all words of 
length $\ell$. One can verify the following useful lemma.
\begin{lemma}\label{lem:maxind}
Let \err be a channel and
let $C\sse\al^\ell$ be an \err-detecting block code.
\begin{enumerate}
\item 
$\mind{C}{\err} = 1$ if and only if 
$C$ is maximal $\err$-detecting.
\item
Assuming that words are chosen uniformly
at random from $\al^\ell$, the maximality index is the probability that a randomly chosen word $w$ of length $\ell$ cannot be added into $C$ preserving its being \err-detecting, that is, 
\[
\mind{C}{\err}\> =\>\prob{w \mbox{ cannot be added into $C$}}.
\]
\end{enumerate}
\end{lemma}
\begin{proof} 
The first statement follows from Definition~\ref{def:maxind} and condition~(\ref{eq:maxed}). The second statement follows when we note that the event that a randomly chosen word $w$ from $\lblock$ cannot be added into $C$ is the same as  the event that $w\in\lblock\cap(\err\lor\erri)(C)$.
\end{proof}
%
%

\section{Generating Error Control Codes}\label{sec:main}
We turn now our attention to algorithms processing channels
and sets of words. A set of words is called a \emph{language},
with a block code being a particular example of language.
A powerful method of representing languages is via finite 
automata \cite{FLhandbookI}. 
A (finite) \emph{automaton} $\aut$ is a 5-tuple 
$(S,\al,I,T,F)$ as in the case of a channel, but each transition
has only an input label, that is, it is of the form
$(s,x,t)$ with $x$ being one alphabet symbol or the empty word $\ew$. The \emph{language accepted by} $\aut$ is denoted
by $\lang{\aut}$ and consists of all words formed by concatenating the labels in any path from an initial to a final state. 
The automaton is called \emph{deterministic}, or \emph{DFA} for short, if $I$ consists of a single state, there are no transitions with label $\ew$, and there are no two distinct transitions with same labels going out of the same state.
Special cases of automata are 
constraint systems in which normally all states 
are final 
(pg 1635--1764 of \cite{Coding:1998}), and trellises. A \emph{trellis} is
an automaton accepting a
block code, and has one initial and one final state 
(pg 1989--2117 of \cite{Coding:1998}).
In the case of a
trellis $\aut$ we talk about the \emph{code represented by} $\aut$, and we denote it as $\code{\aut}$, which is equal to 
$\lang{\aut}$. 
\pnsi
For computational complexity considerations, the \emph{size} $|\autm|$ of a finite state
machine (automaton or transducer) $\autm$ is the number of states
plus the sum of the sizes of the transitions. The size of a transition is 1 plus the length of the label(s) on the
transition. We assume that the alphabet $\al$ is small
so we do not include its size in our estimates.
\pnsi
An important operation between an
automaton $\aut$ and a transducer $\trt$, here denoted by `$\filter$', returns an
automaton $(\aut\filter\trt)$ 
that accepts the set of 
all possible outputs of $\trt$ when the input 
is any word from $\lang{\aut}$, that is,
\[
\quad\lang{\aut\filter\trt} = \trt(\lang{\aut}).
\]
\begin{remark}\label{rem:product}
We recall here the construction of $(\aut\filter\trt)$ from given $\aut=(S_1,A,I_1,T_1,F_1)$ and $\trt=(S_2,A,I_2,T_2,F_2)$, where we assume that $\aut$ contains no transition with label $\e$. 
First, if necessary, we convert $\trt$ to standard form.
Second, if $\trt$ contains any transition whose input label is $\e$, then we add into $T_1$ transitions $(q,\e,q)$, for all states $q\in S_1$. 
 Let $T_1$ denote now the updated set of transitions. Then, we construct the automaton
\[
\autb\>=\>(S_1\times S_2,\, A,\, I_1\times I_2,\, T,\, F_1\times F_2)
\]
such that $((p_1,p_2),y,(q_1,q_2))\in T$, exactly when there are transitions $(p_1,x,q_1)\in T_1$ and $(p_2,x/y,q_2)\in T_2$.
The above construction can be done in time $O(|\aut||\trt|)$ and the size of $\autb$ is $O(|\aut||\trt|)$. The required automaton $(\aut\filter\trt)$ is the trim version of $\autb$, which can be computed in time $O(|\autb|)$. (The trim version of an automaton $\autm$ is the automaton resulting when we remove any states of $\autm$ that do not occur in some path from an initial to a final state of $\autm$.)
\end{remark}
%
%
\begin{center}
\begin{figure}[ht]
\begin{center}
\parbox{0.85\textwidth}
{
\hspace*{0.3\algoindent} 
\nextword\, ($\chann, \aut, \ff, \ef$)
\pssn \hspace*{\algoindent}
\autb\, := $(\aut\filter(\chann\lor\chann^{-1}))$;
\\ \hspace*{\algoindent}
$n$ := 1 + $\Big\lfloor\,1/\Big(4\ef(1-\ff)^2\Big)\,\Big\rfloor$;
\\ \hspace*{\algoindent}
$\ell$ := the length of the words in $\code{\aut}$;
\\ \hspace*{\algoindent}
tr := 1;
\\ \hspace*{\algoindent}
while (tr \mle $n$):
\\\hspace*{1.7\algoindent} 
$w$ := $\pickfrom(\al,\ell)$;
\\ \hspace*{1.7\algoindent} 
if ($w$ not in $\lang{\autb}$) return $w$;
\\ \hspace*{1.7\algoindent} 
tr := tr+1;
\\ \hspace*{\algoindent}
return \none;
}
\caption{Algorithm \nextword---see Theorem~\ref{th:nonmaxw}.}\label{fig:algo1}
\end{center}
\end{figure}
\end{center}
%
%
Next we present our randomized algorithms---we use 
\cite{MiUp:2005} as 
reference for basic concepts. 
We assume that we have available to use in our algorithms
an ideal method $\pickfrom(\al,\ell)$ that chooses 
uniformly at random
a word in $\al^\ell$. 
A randomized algorithm $R(\cdots)$ 
with specific values for its parameters can be viewed as
a random variable whose value is whatever value is returned by 
executing $R$ on the specific values.
\begin{theorem}\label{th:nonmaxw}
Consider the algorithm \nextword\, in Fig.~\ref{fig:algo1}, which
takes as input a channel $\err$, a trellis $\aut$ accepting an \err-detecting code,
and two numbers $\ff,\ef\in[0,1]$.
\begin{enumerate}
\item
The algorithm either returns a word $w\in\al^\ell\setminus\code{\aut}$ such that the code 
$\code{\aut}\cup\{w\}$ is 
$\chann$-detecting, or it returns \none. 
\item
If $\code{\aut}$ is not $\ff$-maximal $\chann$-detecting, then
$$
\prob{\mbox{\nextword\, returns \none}}< \ef.
$$
\item
The time complexity of \nextword\, is 
\[O\Big(\ell|\aut||\chann|\Big/(\ef(1-\ff)^2)\Big).\]
\end{enumerate}
\end{theorem}
\newcommand\Cnt{\ensuremath{\mathtt{Cnt}}}
\newcommand\Tr{\ensuremath{\mathtt{Tr}}}
\begin{proof}
The first statement follows from statement~(\ref{eq:addword})
in the previous section, as any $w$ returned by the algorithm is not
in $(\err\lor\erri)(\code{\aut})$. For the second statement,
suppose that the  code $\code{\aut}$ is not $\ff$-maximal $\chann$-detecting. 
Let  \Cnt\, be the random variable whose value is the value of  tr $-$ 1 at the end of execution of the randomized algorithm \nextword. Then, \Cnt\, counts the number of words that are in $\code{\aut}$ out of $n$ randomly chosen words $w$. 
Thus \Cnt\, is binomial: the number of successes 
(words $w$  in $\lang{\autb}$) in $n$ trials. So $E(\Cnt) = np$, 
where $p=\prob{w\in\lang{\autb}}$. 
By the definition of $n$ in \nextword, 
we get $1/(4n(1-\ff)^2)<\epsilon$.
Now consider the Chebyshev inequality, 
$
\prob{|X-E(X)|\ge a}\le \sigma^2/a^2,
$
where $a>0$ is arbitrary and $\sigma^2$
is the variance of some random variable $X$. For $X=\Cnt$
the variance is $np(1-p)$, and we get 
\[
\prob{\,|\Cnt/n-p|\ge 1-\ff\,}<\ef,
\]
where we  used $a=n(1-\ff)$ and the fact that 
$p(1-p)\le1/4$. 
\pnsi
Using Lemma~\ref{lem:maxind} and the assumption that
$\code{\aut}$ is not $\ff$-maximal, we have that 
$\mind{\code{\aut}}{\err}<\ff$, which implies
$\prob{w\in\lang{\autb}}<\ff$; 
hence, $p<f$. Then
\begin{align*}
\prob{\mbox{ \nextword\, returns \none\;}} 
&=  \prob{\Cnt=n} =\\
\prob{\Cnt/n=1}
&= \prob{\Cnt/n-p=1-p} \le \\
\prob{|\Cnt/n-p|\ge1-p}
&\le \prob{|\Cnt/n-p|\ge1-\ff}\>
< \ef,
\end{align*}
as required.
\par
For the third statement, we use standard results from automaton theory, \cite{FLhandbookI}, and Remark~\ref{rem:product}. In particular,
computing $\autb$ can be done in time 
$O(|\aut|\cdot|\chann|)$  such that 
$|\autb|=O(|\aut|\cdot|\chann|)$.
Testing whether $w\in\lang{\autb}$ can be done in time 
$O(|w||\autb|)=O(\ell|\autb|)$.
Thus, the algorithm works in time
$
O(\ell|\aut||\chann|\,/(\ef(1-\ff)^2)).
$
\end{proof}
\begin{remark}
We mention the important observation that one can modify
the algorithm \nextword\, by removing the construction of
\autb\, and replacing the `if' line in the loop with
\pnsi 
\quad if ($\code{\aut}\cup\{w\}$ is $\err$-detecting) return $w$;
\pnsn
While with this change the output would still be correct,
the time complexity  of the algorithm would increase
to $O\Big(|\aut|^2|\chann|\Big/\big(\ef(1-\ff)^2\big)\Big)$. 
This is because testing whether $\lang{\autv}$ is
$\err$-detecting, for any given automaton $\autv$ and 
channel $\err$, can be done in time $O(|\autv|^2|\chann|)$, and
in practice $|\autv|$ is much larger than $\ell$.
\end{remark}
%
%
In Fig.~\ref{fig:algo2}, we present the main algorithm
for adding new words into a given deterministic trellis $\aut$.
%
%
\begin{figure}[ht]
\begin{center}
\parbox{0.85\textwidth}{
\mbox{}\\
\hspace*{0.3\algoindent} 
\makecode\, ($\chann, \aut, N$)
\pssn \hspace*{\algoindent} 
$W$ := empty list;\quad \autc := \aut
\\ \hspace*{\algoindent}
cnt := 0; \quad  more := True;
\\ \hspace*{\algoindent}
while (cnt \mlt $N$ \texttt{and}\; more)
\\ \hspace*{1.75\algoindent} 
$w$ := \nextword\,($\chann, \autc, \fv, \ev$);
\\ \hspace*{1.75\algoindent} 
if ($w$ is \none) more := False;
\\ \hspace*{1.75\algoindent} 
else \{add $w$ to \autc\, and to $W$; \; cnt := cnt+1;\}
\\ \hspace*{\algoindent} 
 return \autc,\, $W$;
}
\caption{Algorithm \makecode---see Theorem~\ref{th:main}.
The trellis $\aut$ can be omitted
so that the algorithm would start with an empty set of
codewords. In this case,
however, the algorithm would require as extra input
the codeword length $\ell$ and the desired alphabet 
$\al$. We used the fixed values 0.95 and 0.05, as they
seem to work well in practical testing.
}\label{fig:algo2}
\end{center}
\end{figure}
\begin{remark}
In some sense, algorithm \makecode\, generalizes to arbitrary channels the idea used in the proof of the well-known Gilbert-Varshamov bound~\cite{MacWilliams:Sloane} for the largest possible block code $M\sse\lblock$ that is $\chsub k$-correcting, for some number $k$ of substitution errors. In that proof, a word can be added into the code $M$ if the word is outside of the union of the ``balls'' $\chsub{2k}(u)$, for all $u\in M$. In that case, we have that $\chsub k^{-1}=\chsub k$ and $(\chsub k^{-1}\circ\chsub k)=\chsub{2k}(u)$. The present algorithm
adds new words $w$ to the constructed trellis $\autc$ such that each new word $w$ is outside of the ``union-ball'' $(\err\lor\erri)(\code{\autc})$.
\end{remark}
\begin{theorem}\label{th:main}
Algorithm \makecode\, in Fig.~\ref{fig:algo2} 
takes as input a channel $\err$, 
a deterministic trellis $\aut$ of some length $\ell$, and an integer $N>0$ such that
the code $\code{\aut}$ is \err-detecting, and
returns a deterministic 
trellis $\autc$
and a list $W$ of words 
such that the following statements hold true:
\begin{enumerate}
\item 
$\code{\autc}=\code{\aut}\cup W$ and $\code{\autc}$ is 
\chann-detecting, 
\item
If $W$ has less than $N$ words, then 
either 
$\mind{\code{\autc}}{\err}\ge\fv$
or
the probability that a randomly chosen word from
$\al^\ell$ can be added in $\code{\autc}$ is $<\ev$.
\item
The algorithm runs in time 
$O\Big(\ell N|\err||\aut|+\ell^2N^2|\err|\Big)$.
\end{enumerate}
\end{theorem}
\begin{proof}
Let $\autc_i$ be the value of the trellis $\autc$ at the end of the $i$-th iteration
of the while loop.
The first statement follows from Theorem~\ref{th:nonmaxw}: any word $w$ returned by \nextword\, is such that $\code{\autc_i}\cup\{w\}$ is \err-detecting.
For the second statement, assume that, at the end of 
execution, $W$ has $<N$ words and $\code{\autc}$
is not \fp-maximal. By the previous theorem, this means 
that the
random process $\nextword(\err,\autc,\fv,\ev)$
returns \none\, with probability $<\ev$,
as required.
For the third statement, as  the loop in the algorithm \nextword\, performs a fixed number of iterations (=\nv), 
we have that the
cost of \nextword\, is $O(\ell|\autc_i||\err|)$. 
The cost of adding a new word $w$ of length $\ell$ 
to $\autc_{i-1}$ is
$O(\ell)$ and increases its size by $O(\ell)$, 
so each $\autc_i$ is of size $O(|\aut|+i\ell)$.
Thus, the cost of the $i$-th iteration of the while loop in \makecode\, 
is $O(\ell|\err|(|\aut|+i\ell))$. As there are up to $N$
iterations the total cost is
$$
\sum_{i=1}^NO\Big(\ell|\err|\cdot(|\aut|+i\ell)\Big)\>=\>
O\Big(\ell N|\err||\aut|+\ell^2N^2|\err|\Big).
$$
\end{proof}
\begin{remark}\label{rem:time}
In the algorithm \makecode, attempting to add only one word into $\code{\aut}$ (case of $N=1$), requires time $O(\ell |\err||\aut|+\ell^2|\err|)$, which is of polynomial magnitude. This case is equivalent to testing whether $\code{\aut}$ is maximal \err-detecting, which is shown to be a hard decision problem in Theorem~\ref{th:coNP}.
\end{remark}

\begin{remark}
In the version of the algorithm \makecode\, where the initial trellis $\aut$ is
omitted, the time complexity is
$O(\ell^2N^2|\err|)$. We also note that the algorithm would work with the same time complexity if the given trellis $\aut$ is not deterministic. In this case, however, the resulting trellis would not be  (in general) deterministic either.
\end{remark}

\section{Why not Use a Deterministic Algorithm}\label{sec:coNP}
Our motivation for considering randomized algorithms is that the \emph{embedding problem} is computationally hard: given a deterministic trellis $\autd$ and a channel \err, compute (using a deterministic algorithm) a  trellis that represents a maximal \err-detecting  code containing $\code{\autd}$. By computationally hard, we mean that a decision version of the embedding problem is coNP-hard. This is shown next. 

\begin{theorem}\label{th:coNP}
	The following decision problem is coNP-hard.
\begin{description}
  \item[\textit{Instance:}] deterministic trellis $\autd$ and channel \err.
  \item[\textit{Answer:}] whether $\code{\autd}$ is maximal \err-detecting.
\end{description}
\end{theorem}
\begin{proof}
Let us call the decision problem in question \maxed, and let \fullblock
be the problem of deciding whether a given trellis over the alphabet $\al_2=\{0,1\}$ with no $\e$-labeled transitions accepts $\al_2^\ell$, for some $\ell$.
The statement is a logical consequence of the following claims.
\pnsn  \textit{Claim 1:} \fullblock is coNP-complete.
\\  \textit{Claim 2:} \fullblock is polynomially reducible to \maxed.
\par
The first claim follows from the proof of the following fact on page 329 of \cite{LePa1998}:
Deciding whether two given star-free regular expressions over $\al_2$ are inequivalent is an NP-complete problem. Indeed, in that proof the first regular expression can be arbitrary, but the second regular expression represents the language $\al^\ell$, for some positive integer $\ell$. Moreover, converting
a star-free regular expression to an acyclic automaton with no $\e$-labeled transitions is a polynomial time problem.
\par
For the second claim, consider any trellis $\aut=(S,A_2,s,T,F)$ with no $\e$-labeled transitions in $T$. We need to construct in polynomial time an instance $(\autd,\err)$ of \maxed such that $\aut$ accepts $\al_2^\ell$ if and only if $\code{\autd}$ is a maximal \err-detecting block code of length $\ell$. The rest of the proof consists of 5 parts: construction of deterministic trellis $\autd$ accepting words of length $\ell$, construction of \err, facts about $\autd$ and \err, proving that $\code{\autd}$ is \err-detecting, proving that $\aut$ accepts $\al_2^\ell$ if and only if $\code{\autd}$ is maximal \err-detecting.
\pnsi
\textit{Construction of $\autd$:}
Let $\al$ be the alphabet $A_2\dot\cup T$, where $T$ is the set of transitions of $\aut$. The required deterministic trellis $\autd$ is any deterministic trellis accepting $\al^\ell\sm\al_2^\ell$, that is,
\[
\code{\autd}\>=\>\al^\ell\sm\al_2^\ell.
\] 
This can be constructed, for instance, by making deterministic trellises $\autd_1$ and $\autd_2$ accepting, respectively, $\al^\ell$ and $\al_2^\ell$, and then intersecting $\autd_1$ with the complement of $\autd_2$. Note  that any word in $\code{\autd}$ contains at least one symbol in $T$. 
\pnsi
\textit{Construction of $\err$:}
This is of the form $(\err_1\lor\err_2)$ as follows. The transducer $\err_2$ has only one state $s$ and transitions $(s,\alpha/\alpha,s)$, for all $\alpha\in \al$, and realizes the identity relation $\{(x,x)\mid x\in\al^*\}$. Thus, we have that $\err_2(x)=\{x\}$, for all words $x\in\al$. The transducer $\err_1=(S,A,s,T'',F)$ is such that $T''$ consists of exactly the transitions $(p,(p,a,q)/a,q)$ for which $(p,a,q)$ is a transition of $\aut$. 
\pnsi
\textit{Facts about $\autd$ and \err:} The following facts are helpful in the rest of the proof. Some of these facts refer to the deterministic trellis $\autd_1=(S,T,s,T_1,F)$ resulting by omitting the output parts of the transition labels of $\err_1$, that is, $(p,(p,a,q),q)\in T_1$ exactly when $(p,(p,a,q)/a,q)\in T''$. Then, $\code{\autd_1}\sse(\al\sm\al_2)^\ell\sse\code{\autd}$. 
    \par \noindent
      F0: $\lang{\autd_1\filter\err_1}=\lang{\aut}$.
    \\F1: The domain of $\err_1$ is $\code{\autd_1}$, a subset of $(A\sm\al_2)^\ell$.
    \\F2: If $v\in\err_1(u)$ then $v\in\al_2^\ell$ and $v\not=u$.
    \\F3: $\err_1(\code{\autd})=\lang{\aut}$.
    \\F4: $\erri_1(\code{\autd})=\emptyset$.
    \par
    For fact F0, note that the product construction described in Remark~\ref{rem:product} produces in $(\autd_1\filter\err_1)$ exactly the transitions $((p,p),a,(q,q))$, where $(p,a,q)$ is a transition in $\aut$, by matching  any transition $(p,(p,a,q),q)$ of $\autd_1$ only with the transition $(p,(p,a,q)/a,q)$ of $\err_1$. 
    Fact F1 follows by the construction of $\err_1$ and the definition of $\autd_1$: in any accepting computation of $\err_1$, the input labels appear in an accepting computation of $\autd_1$ that uses the same sequence of states.
    F3 is shown as follows: As the domain of $\err_1$ is $\code{\autd_1}$ and $\code{\autd_1}\sse\code{\autd}$, we have that $\err_1(\code{\autd})=\err_1(\code{\autd_1})$, which is $\lang{\aut}$ by F0. Fact F4 follows by noting that the domain of $\erri_1$ is a subset of $\al_2^\ell$ but $\code{\autd}$ contains no words in $\al_2^\ell$.
\pnsi
\textit{$\code{\autd}$ is \err-detecting:}
Let $u,v\in\code{\autd}$ such that $v\in\err(u)=\err_1(u)\cup\{u\}$.
We need to show that $v=u$, that is, to show that $v\notin\err_1(u)$. Indeed, if $v\in\err_1(u)$ then $v\in\al_2^\ell$, which contradicts  $v\in\code{\autd}=\al^\ell\sm\al_2^\ell$.
\pnsi
\textit{$\aut$ accepts $\al_2^\ell$ if and only if $\code{\autd}$ is maximal \err-detecting:} By statement~(\ref{eq:maxed}) we have that $\code{\autd}$ is maximal \err-detecting, if and only if $(\err\lor\erri)(\code{\autd})=\al^\ell$. We have:
\begin{eqnarray*}
(\err\lor\erri)(\code{\autd}) &=& \code{\autd}\cup\err_1(\code{\autd})\cup\erri_1(\code{\autd})\\
&=& (\al^\ell\sm\al_2^\ell)\cup \lang{\aut}\cup\emptyset\\
&=&(\al^\ell\sm\al_2^\ell)\>\dot\cup\> \lang{\aut}.
\end{eqnarray*}
Thus, $\code{\autd}$ is maximal \err-detecting, if and only if $\lang{\aut}=\al_2^\ell$, as required.
\end{proof}

\section{Implementation and Use}\label{sec:implem}
All main algorithmic tools have been implemented over
the years in the Python package FAdo \cite{Fado,AAAMR:2009,KMMR:2015}.
Many aspects of the new module FAdo.codes are presented in~\cite{KMMR:2015}.
Here we present methods of that module pertaining to generating codes.
\pnsi
Assume that the string \texttt{d1} contains a description of
the transducer $\chdel1$ in FAdo format. In particular, \texttt{d1} begins with
the type of FAdo object being described, the final states, and the initial states (after the character \texttt{*}). Then, \texttt{d1} contains the
list of transitions, with each one of the form ``$s$ $x$ $y$ $t$\texttt{\textbackslash n}'', where `\texttt{\textbackslash n}' is the new-line character. This shown in the following Python script.
\begin{verbatim}
    import FAdo.codes as codes
    d1 = '@Transducer 0 2 * 0\n'
      '0 0 0 0\n0 1 1 0\n0 0 @epsilon 1\n'
      '0 1 @epsilon 1\n1 0 0 1\n1 1 1 1\n'
      '1 @epsilon 0 2\n1 @epsilon 1 2\n'
    pd1 = codes.buildErrorDetectPropS(d1)
    a = pd1.makeCode(100, 8, 2)
    print pd1.notSatisfiesW(a)
    print pd1.nonMaximalW(a, m)
    s2 = ...string for transducer sub_2
    ps2 = codes.buildErrorDetectPropS(s2)
    pd1s2 = pd1 & ps2
    b = pd1s2.makeCode(100, 8, 2)
\end{verbatim}
The above script uses the string \texttt{d1} to create the object \texttt{pd1}
representing the $\chdel1$-detection property over the alphabet \{0,1\}. Then, 
it constructs an automaton \texttt{a} representing
a $\chdel1$-detecting block code of length $8$ 
with up to $100$ words over the 2-symbol alphabet \{0,1\}. 
The method 
\texttt{notSatisfiesW(a)} tests whether the code
$\code{\mathtt{a}}$ is $\chdel1$-detecting and returns
a witness of non-error-detection (= pair of codewords $u,v$ with $v\in\chdel1(u)$), or (\texttt{None}, \texttt{None})---of course, in the above example 
it would return (\texttt{None}, \texttt{None}). The method
\texttt{nonMaximalW(a, m)} tests whether the code
$\code{\mathtt{a}}$ is maximal $\chdel1$-detecting and returns
either a word $v\in\lang{\mathtt{m}}\setminus\code{\mathtt{a}}$ such that $\code{\mathtt{a}}\cup\{v\}$ is $\chdel1$-detecting, or \texttt{None} if $\code{\mathtt{a}}$ is already maximal.
The object \texttt{m} is any automaton---here it is
the trellis representing $\al^\ell$.
This method is used only for  small codes, as 
in general
the maximality problem is algorithmically hard (recall Theorem~\ref{th:coNP}), which
motivated us to consider the randomized version \nextword\,
in this paper. 
For any channel $\err$ and trellis \texttt{a}, the method \texttt{notSatisfiesW(a)} can be made to work in
time $O(|\err||\mathtt{a}|^2)$, which is of polynomial 
complexity. 
The operation `\verb1&1' combines error-detection properties.
Thus, the second call to \makecode\, constructs a code that
is $\chdel1$-detecting and $\chsub2$-detecting (=$\chsub1$-correcting).

\section{More on Channel Modelling, Testing}\label{sec:further}
In this section, we consider further examples of channels and show how operations on channels can result in new ones. We also show the results of testing our 
codes generation algorithm for several different channels. 
\begin{remark}\label{rem:errspecs}
We note that the definition of error-detecting (or error-correcting) block code $C$ is trivially extended to any language $L$, that is, one replaces in Definition~\ref{def:errcontrol} `block code $C$' with `language $L$'. 
	Let $\err, \err_1,\err_2$ be  channels. By Definition~\ref{def:errcontrol} and using standard logical arguments, it follows that 
\begin{enumerate}
    \item\label{rem:item:chboth} 	
	$L$ is $\err_1$-detecting and $\err_2$-detecting, if and only if $L$ is $(\err_1\lor\err_2)$-detecting;
    \item\label{rem:item:chvsinv}
    $L$ is $\erri$-detecting, if and only if it is $\err$-detecting, if and only if it is $(\erri\lor\err)$-detecting.	
\end{enumerate}	
\end{remark}
The inverse of $\chdel1$ is $\chins1$ and is shown in 
Fig.~\ref{fig:ed}, where recall it results by simply exchanging
the order of the two words in all the labels in $\chdel1$. 
By statement~2 of the above remark, the 
$\chdel1$-detecting codes are the same as the $\chins1$-detecting 
ones, and the same as the $(\chdel1\lor\chins1)$-detecting ones---this is shown in \cite{PAF:2013} as well. 
The method of using transducers
to model channels is quite general and one can give many more
examples of past channels as transducers, as well as channels not studied before. Some further examples are shown in the next figures, Fig.~\ref{fig:ed}-\ref{fig:solid}. 
\par
One can go beyond the
classical error control properties and define certain synchronization properties via transducers. Let $\ofw$ be the set of all overlap-free words, that is, all words $w$ such that a proper and nonempty prefix of $w$ cannot be a suffix of $w$. A block code $C\sse\ofw$ is a \emph{solid code} if any proper and nonempty prefix of a $C$-word cannot be a suffix of a $C$-word. For example, \{0100, 1001\} is not a block solid code, as 01 is a prefix and a suffix of some codewords and 01 is nonempty and a proper prefix (shorter than the codewords). Solid codes can also be non-block codes by extending appropriately the above definition~\cite{Shyr:book} (they are also called \emph{codes without overlaps} in~\cite{Levenshtein:73}). The transducer $\chov$ in Fig.~\ref{fig:solid} is such that any block code $C\sse\ofw$
is a solid code, if and only if $C$ is an `$\chov$-detecting' block code. We note that solid codes have instantaneous synchronization capability (in particular all solid codes are comma-free codes)
as well as synchronization in the presence of noise~\cite{JurgYu90}. 
\begin{figure}[ht!]
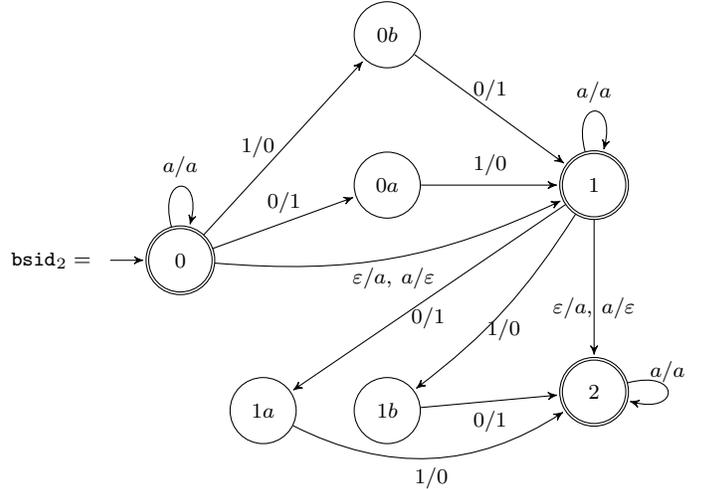

\begin{transducer}
\node [state,initial,accepting] (q0) {$0$};
	\node [node distance=1.00cm,left=of q0,anchor=east] 
	      {$\chbsid2=\>$}; 
	\node [above of=q0,node distance=1.00cm](tmp1){};
	\node [state,right of=tmp1] (q0a) {$0a$};
	\node [state,above of=q0a,node distance=2.00cm] (q0b) {$0b$};
	\node [state,accepting,right of=q0a] (q1) {$1$};
	\node [below of=q0,node distance=2.00cm](tmp2){};
	\node [state,right of=tmp2,node distance=1.10cm] (q1a) {$1a$};
	\node [state,below of=q0a,node distance=3.00cm] (q1b) {$1b$};
	\node [state,accepting,below of=q1] (q2) {$2$};
\path 
		(q0) edge [bend right=15] node [below] {$\e/\symb,\>\symb/\e$} (q1)
		(q0) edge node [above] {$0/1$} (q0a)
		(q0a) edge node [above] {$1/0$} (q1)
		(q0) edge node [left] {$1/0$} (q0b)
		(q0b) edge node [above] {$0/1$} (q1)
		(q1) edge node [below] {$0/1$} (q1a)
		(q1) edge [bend left=10] node [below] {$1/0$} (q1b)
		(q1a) edge [bend right=30] node [below] {$1/0$} (q2)
		(q1b) edge node [below] {$0/1$} (q2)
		(q1) edge node [below] {$\ew/\symb,\>\symb/\ew$} (q2)
		(q0) edge [loop above] node [above] {$\symb/\symb$} ()
		(q1) edge [loop above] node [above] {$\symb/\symb$} ()
		(q2) edge [loop right] node [above] {$\symb/\symb$} ();
\end{transducer}
\begin{center}
\caption{The channel specified by $\chbsid{2}$ allows up to two errors in the input word. Each of these errors can be a deletion, an insertion, or a bit shift: a 10 becomes 01, or a 01 becomes 10. The alphabet is \{0, 1\}.
}
\label{fig:ed}
\end{center}
\end{figure}
\begin{figure}[ht!]
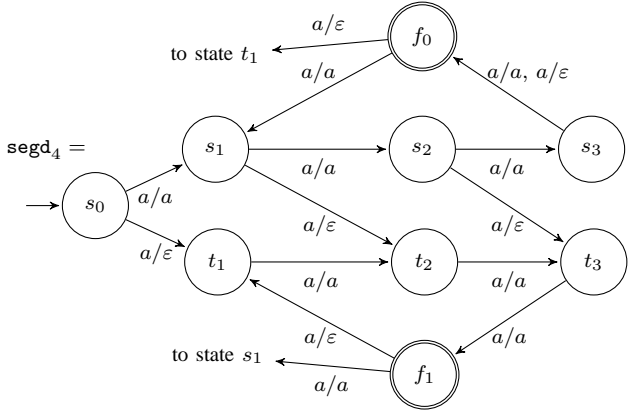

\begin{transducer}
	\node [state,initial] (q0) {$s_0$};
	\node [node distance=0.75cm,above=of q0,anchor=east] (n0)
	      {$\chsegd4=$};
	\node [node distance=0.75cm,below=of q0,anchor=east] (n1)
	      {};
	\node [state,node distance=2.25cm,right of=n0] (p1) {$s_1$};
	\node [state,node distance=1.75cm,right of=n1] (q1) {$t_1$};
	\node [state,right of=p1] (p2) {$s_2$};
	\node [state,right of=q1] (q2) {$t_2$};
	\node [node distance=1.25cm,above=of p1] (nn0)
	      {to state $t_1$};
	\node [node distance=1.25cm,below=of q1] (nn1)
	      {to state $s_1$};
	\node [node distance=2.25cm,state,right of=p2] (p3) {$s_3$};
	\node [node distance=2.25cm,state,right of=q2] (q3) {$t_3$};
	\node [state,accepting,node distance=1.5cm,above=of p2] (m2)
	      {$f_0$};
	\node [state,accepting,node distance=1.5cm,below=of q2] (n2)
	      {$f_1$};
	\path 
	      (q0) edge node [below] {$\symb/\ew$} (q1)
		  (q0) edge node [below] {$\symb/\symb$} (p1)
	      (q1) edge node [below] {$\symb/\symb$} (q2)
		  (p1) edge node [below] {$\symb/\symb$} (p2)
		  (p1) edge node [below] {$\symb/\ew$} (q2)
	      (q2) edge node [below] {$\symb/\symb$} (q3)
		  (p2) edge node [below] {$\symb/\symb$} (p3)
		  (p2) edge node [below] {$\symb/\ew$} (q3)

	    (p3) edge node [above] {$\quad\quad\symb/\symb,\,\symb/\ew$} (m2)
		  (q3) edge node [below] {$\symb/\symb$} (n2)
	    (m2) edge node [above] {$\symb/\symb$} (p1)
		(n2) edge node [below] {$\symb/\ew$} (q1)
	    (m2) edge node [above] {$\symb/\ew$} (nn0)
		(n2) edge node [below] {$\symb/\symb$} (nn1)
        ;
\end{transducer}
\begin{center}
\caption{
Transducer for the segmented deletion channel of~ \cite{LiMi:2007} with parameter $b=4$. 
In each of the length $b$ consecutive segments of the 
input word, at most one deletion error occurs. The length of the
input word is a multiple of $b$. By Lemma~\ref{lem:edvsec}, 
$\chsegd4$-correction is equivalent to 
$(\chsegd4^{-1}\circ\chsegd4)$-detection.
\label{fig:segmented}
}
\end{center}
\end{figure}
\begin{figure}[ht!]
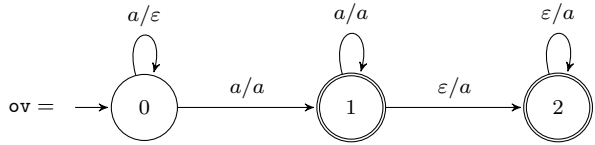

\begin{transducer}
	\node [state,initial] (q0) {$0$};
	\node [node distance=1.00cm,left=of q0,anchor=east] 
	      {$\chov=\>$};
	\node [state,accepting,right of=q0] (q1) {$1$};
	\node [state,accepting,right of=q1] (q2) {$2$};
	\path (q0) edge [loop above] node [above] {$\symb/\e$} ()
		(q0) edge node [above] {$\symb/\symb$} (q1)
		(q1) edge node [above] {$\e/\symb$} (q2)
		(q1) edge [loop above] node [above] {$\symb/\symb$} ()
		(q2) edge [loop above] node [above] {$\e/\symb$} ();
\end{transducer}
\begin{center}
\caption{This input-preserving transducer  deletes a prefix of the input word (a possibly empty prefix) and then  inserts a possibly empty suffix at the end of the input word. 
}
\label{fig:solid}
\end{center}
\end{figure}
\pnsi
For $\ef=\ev$ and $\ff=\fv$, the value of 
$n$ in \nextword\, is \nv. We performed several executions
of the algorithm \makecode\, on various channels using
$n=\nv$, no initial trellis, 
and alphabet $\al=\{0,1\}$.
%
%
\begin{table*}[t]
\begin{center}
\begin{tabular}{l|ccccc}
\toprule
$N\!\!=$, \!$\ell\!\!=$, end=\!& $\chid2$ & $\chdel1$ & $\chsub2$ &$\chbsid{2}$ & $\chov$\\  \midrule 
$100,\>\> 8,\,$   & $18, 20, 23$ & $37, 42, 51$  & $15, 16, 18$ & $17,19,21$ & $01,07,08$\\  
$100,\>\>  7,\,$  & $10, 12, 13$ & $20, 23, 28$ & $09, 10, 13$ & $11,11,13$ & $03,04,05$\\  
$100,\>\>  8,\,\>\>\>\>$ 1 & $11, 13, 14$ & $39, 50, 64$ & $09, 10, 11$ & $09,12,13$ & $01,05,06$\\  
$100,\>\>  8,\,\>\>$ 01& $06, 07, 08$ & $64, 64, 64$ & $04, 06, 08$ & $06,07,09$ & $01,04,05$\\  
$500,\>\, 12,\,$   & $177,182,188$ & $500,500,500$ & $148,157,162$ &$169, 173, 178$ & $51,59,63$\\ 
$500,\>\, 13,\,$   & $318,327,334$ &  & $272,273,278$ & $302,303,309$ & $43,111,120$\\ 
\bottomrule  
\end{tabular}
\end{center}
\end{table*}
%
In the above table, the first column gives the values of 
$N$ and $\ell$, and if present and nonempty, the pattern that all
codewords should end with (1 or 01). 
For each entry in an `$N=100$' row, we executed \makecode\, 21 times 
and reported 
smallest, median, and largest sizes of the 21 generated codes.
For $N$ = 500, we reported the same figures by executing the
algorithm 5 times.
For example, the entry 37,42,51 corresponds to executing
\makecode\, 21 times for $\err=\chdel1$, $\ell=8$, 
end = $\ew$. The entry 64,64,64 corresponds to the
systematic code of \cite{PAF:2013} whose codewords end with 01,
and any of the $64$ 6-bit words can be used in 
positions 1--6. The entry  for
`$\ell=7$, end = $\ew$, $\err=\chsub2$'
corresponds to 2-substitution error-detection which is
equivalent to 1-substitution error-correction. Here the
Hamming code of length 7 with $16$ codewords has a 
maximum number of codewords for this  length.
Similarly, the entry for `$\ell=7$,  $\err=\chid2$'
corresponds to 2-synchronization error-detection which is
equivalent to 1-synchronization error-correction. 
Here the Levenshtein code \cite{Levenshtein:66:en}  
of length 8 has 30 codewords.
We recall that a maximal code is
not necessarily maximum, that is, having the largest
possible number of codewords, for given \err and $\ell$. 
It seems maximum codes are rare, but there are
many random maximal ones having lower rates. 
The $\chdel1$-detecting code of~\cite{PAF:2013} 
has higher rate than all the random ones generated here.
\par
For the case of block solid codes (last column of the table), we note that the function \pickfrom\, in the algorithm \nextword\, has to be modified as the randomly chosen word $w$ should be in \ofw.
%
%
\section{Conclusions}\label{sec:last}
\pnsi
We have presented a unified method for generating
error control codes, for any rational combination of errors. The method cannot of course
replace innovative code design, but should be 
helpful in computing various examples of codes. The implementation \texttt{codes.py}
is available to anyone for download and use \cite{Fado}.
In the implementation for generating codes, we allow one to specify that generated words only come from a certain desirable subset $M$ of $\al^\ell$, which  is represented by a deterministic trellis. This requires changing the function \pickfrom\, in \nextword\, so that it chooses randomly words from $M$.
There are a few directions for 
future research. One  is to work on the efficiency of the implementations, possibly allowing parallel processing, so as to allow generation of block codes having longer block length.
Another direction is to somehow
find a way to specify that the set of generated codewords  
is a `systematic' code so as to allow efficient 
encoding of information. 
A third direction is to do
a systematic study on how one can map a stochastic
channel $\stoch$, like the binary symmetric channel or
one with memory, to a channel \err (representing a
combinatorial channel), so as the available algorithms
on $\err$ have a useful meaning on $\stoch$
as well.

\bibliographystyle{plain}
\bibliography{refs}

\end{document}